\documentclass[11pt,letterpaper,american]{article}
\usepackage{fullpage} 

\usepackage{mdwlist}

\usepackage{amsmath}
\usepackage{graphics, subfigure, float}
\usepackage{fp, calc,pst-func}
\usepackage{hyperref}

\usepackage[latin1]{inputenc}
\usepackage[american]{babel}
\usepackage[T1]{fontenc} 
\usepackage{fourier}

\usepackage{amscd,amsthm}
\usepackage{comment,verbatim}

\usepackage{pst-all}
\usepackage{pstricks-add}
\newpsobject{showgrid}{psgrid}{subgriddiv=1,griddots=10,gridlabels=6pt}
\usepackage{bm}

\newtheoremstyle{theorem}{1em}{1em}{\slshape}{0pt}{\bfseries}{.}{ }{}
\theoremstyle{theorem}
\newtheorem{theorem}{Theorem}
\newtheorem*{theorem*}{Theorem}

\newtheorem{lemma}[theorem]{Lemma}

\newtheorem*{claim*}{Claim}
\newtheorem*{example*}{Example}
\newtheorem*{lemma*}{Lemma}
\newtheorem*{conjecture*}{Conjecture}

\theoremstyle{remark}

\newtheorem*{remark*}{Remark}

\providecommand{\setZ}{\mathbb{Z}}

\providecommand{\setR}{\mathbb{R}}

\newcommand{\E}{\mathop{\mathbb{E}}}

\psset{linewidth=1pt, arrowsize=6pt}
\floatstyle{ruled}
\newfloat{algorithm}{tbp}{loa}
\floatname{algorithm}{Algorithm}

 \theoremstyle{definition}
 
  \theoremstyle{plain}
  
  \theoremstyle{plain}
  
  \theoremstyle{plain}

\usepackage[displaymath,textmath,sections,graphics, subfigure]{preview} 
\PreviewEnvironment{myproblem}
\PreviewEnvironment{defn} 
\PreviewEnvironment{center} 
\PreviewEnvironment{pspicture} 

        \def\drawRect#1#2#3#4#5{
           \FPeval{\x2}{(#2) + (#4)} 
           \FPeval{\y2}{(#3) + (#5)} 
           \pspolygon[#1](#2,#3)(\x2,#3)(\x2,\y2)(#2,\y2)
        }

\usepackage{calrsfs} 
\DeclareMathAlphabet{\pazocal}{OMS}{zplm}{m}{n}

\makeatother

\title{Lecture Notes on the ARV Algorithm for Sparsest Cut}
\author{Thomas Rothvoss\thanks{Email: {\tt rothvoss@uw.edu}. Supported by NSF grant 1420180 with title ``\emph{Limitations of convex relaxations in combinatorial optimization}'' and an Alfred P. Sloan Research Fellowship.} \vspace{3mm} \\  University of Washington, Seattle}

\begin{document}

\maketitle

\begin{abstract}
One of the landmarks in approximation algorithms is the $O(\sqrt{\log n})$-approximation algorithm
for the \emph{Uniform Sparsest Cut problem} by Arora, Rao and Vazirani from 2004. 
The algorithm is based on a \emph{semidefinite program} that finds an embedding of the nodes 
respecting the triangle inequality. Their core argument shows that a random hyperplane approach will
find two large sets of $\Theta(n)$ many nodes each that have a distance of $\Theta(1/\sqrt{\log n})$
to each other if measured in terms of $\|\cdot \|_2^2$.

Here we give a detailed set of lecture notes describing the algorithm. For the proof of 
the \emph{Structure Theorem} we use a cleaner argument based on expected maxima over $k$-neighborhoods that significantly simplifies
the analysis.  
\end{abstract}


\section{Introduction}

Let $G = (V,E)$ be a complete, undirected graph on $|V| = n$ nodes and let $c : E \to \setR_{\geq 0}$
be a cost function on the edges. For a subset $S \subseteq V$ of nodes, let $\delta(S) := \{ \{ i,j\} \in E \mid |\{ i,j\} \cap S| = 1\}$ be the induced \emph{cut}. We abbreviate $c(\delta(S)) := \sum_{e \in \delta(S)} c(e)$ as the
cost of the cut. 
The \emph{(Uniform) Sparsest Cut} problem is then to find the cut that minimizes the \emph{cost-over-separated-pairs} ratio: 
\[
  \min \left\{ \frac{c(\delta(U))}{|U| \cdot |V \setminus U|} \mid \emptyset \subset U \subset V \right\}.
\]
There is also a non-uniform version of the problem where each pair $i,j \in V$ has an associated
\emph{demand} $d(i,j) \geq 0$ and one aims for the cut minimizing the ratio $c(\delta(S)) / d(\delta(S))$. 
We will now see the celebrated algorithm by Arora, Rao and Vazirani~\cite{DBLP:conf/stoc/AroraRV04}
that finds a $O(\sqrt{\log n})$-approximation in polynomial time. 

For the algorithm we will not try to optimize any constant. 
To fix some notation, we will denote any vector in bold font, like $\bm{v}_i \in \setR^m$. 
If we write $i \sim V$, then we mean that $i$ is a uniform random node from $V$.  
We denote $N(0,1)$ as the \emph{$1$-dimensional Gaussian distribution} with mean 0 and variance 1. 
In particular, a random variable $g \sim N(0,1)$ has \emph{density} $\frac{1}{\sqrt{2\pi}}e^{-x^2/2}$.
If we write $\bm{g} \sim N^m(0,1)$, then we mean that $\bm{g}$ is an \emph{$m$-dimensional Gaussian}.
Recall that the vector $\bm{g} = (g_1,\ldots,g_m)$ can be generated by sampling each coordinate independently
with $g_i \sim N(0,1)$. In reverse, for any pair of orthonormal vectors $\bm{u},\bm{v} \in \setR^m$, 
the inner products $\left<\bm{g},\bm{u}\right>,\left<\bm{g},\bm{v}\right>$ are independently distributed from $N(0,1)$.

\section{A semidefinite program}

Sparsest Cut is an unusual problem in the sense that it minimizes the ratio of two functions.  
Let us assume for the sake of simplicity that we guessed the cost $C^*$ and the size $S^*$ of an optimum cut, 
say with $S^* \leq \frac{n}{2}$. 
Then we define a \emph{semidefinite programming relaxation}
\begin{eqnarray*}
\begin{array}{rclll}
 \displaystyle \sum_{i \in V} \|\bm{v}_i\|_2^2 &=& S^* & & (I) \\
\displaystyle  \sum_{\{i,j\} \in {V \choose 2}} \|\bm{v}_i - \bm{v}_j\|_2^2 &=& S^* \cdot (n-S^*) & & (II) \\
 \displaystyle \sum_{\{i,j\} \in E} c_{ij}\|\bm{v}_i - \bm{v}_j\|_2^2 &=& C^* & & (III) \\
\displaystyle  \| \bm{v}_i - \bm{v}_j\|_2^2 &\leq& \|\bm{v}_i - \bm{v}_k\|_2^2 + \|\bm{v}_k - \bm{v}_j\|_2^2 & \forall i,j,k \in V \cup \{0\} & (IV) \\
 \displaystyle \|\bm{v}_i\|_2^2 &\leq& 1 & \forall i \in V & (V) \\
 \displaystyle \|\bm{v}_{0}\|_2^2 &=& 0 & & (VI)
\end{array}
\end{eqnarray*}
where we use an artificial index $0$ with $\bm{v}_0 := \bm{0}$, so that the triangle inequality 
also holds for the origin.
\begin{lemma}
If there is cut $U^* \subseteq V$ of cost $C^*$ and size $S^*$, then the above SDP has a solution. 
\end{lemma}
\begin{proof}
One could choose 1-dimensional vectors by defining 
\[
  \bm{v}_i := \begin{cases} 1 & \textrm{if }i \in U^* \\ 0 & \textrm{if }i \notin U^*. \end{cases}
\]
Then the only non-trivial case is verifying the \emph{triangle inequalities}
in $(IV)$. These are satisfied by our choice of $\bm{v}_i \in \{ 0,1\}$ 
since if $\|\bm{v}_i - \bm{v}_j\|_2^2 = 1$, then $i$ and $j$ have to be on different sides of the cut $U^*$ and
any node $k$ has to be either not on the side of $i$ or not on the side of $j$. 
\end{proof}
We can solve the semi-definite program (SDP) in polynomial time~\cite{GroetschelLovaszSchrijver88}; let $\{\bm{v}_i\}_{i \in V} \subseteq \setR^m$
be the solution. 
Due to the triangle inequalities (IV) we can define a \emph{metric} $d : V \times V \to \setR_{\geq 0}$
by setting $d(i,j) := \|\bm{v}_i - \bm{v}_j\|_2^2$. Note that while $\|\cdot \|_2$ is always a metric, $\|\cdot \|_2^2$ is
 not a metric on all points sets. For sets of nodes $A,B \subseteq V$ 
we define $d(i,A) := \min\{ d(i,j) : j \in A\}$ and $d(A,B) := \min\{ d(i,j) : i \in A, j \in B\}$.

\section{A ball rounding scheme}

Given that family $\{\bm{v}_i\}_{i \in V}$ of SDP vectors, there are several natural rounding 
procedures that would come to mind. For example one could try the \emph{hyperplane rounding}
that Goemans and Williamson~{\cite{DBLP:conf/stoc/GoemansW94} have used for MaxCut. 
The natural algorithm for Sparsest Cut
would be to take a random Gaussian $\bm{g} \sim N^m(0,1)$ and set $U := \{ i \in V \mid \left<\bm{g},\bm{v}_i\right> \geq 0\}$. Assume for the 
sake of simplicity that we are in the \emph{balanced} case of Sparsest Cut
with $S^* = \Theta(n)$. Then an edge $(i,j) \in E$ has a contribution to the objective function of $\Theta(d(i,j) / n^2)$. On the other hand, the probability that $(i,j)$ is separated is roughly proportional to the
Euclidean distance $\|\bm{v}_i - \bm{v}_j\|_2$ and even if the hyperplane generates perfectly balanced cuts,  
the expected contribution of an edge $(i,j)$ to the hyperplane cut would be $\Theta(\sqrt{d(i,j)} / n^2)$. 
In other words, short edges would be separated far too likely.

The second best idea that one might have, would be to select a node $i \in V$ and take a random cut $U := \{ j \in V \mid d(i,j) \leq r\}$ where $r \sim [0,1]$. Now every edge is cut with a probability not exceeding $d(i,j)$. 
On the other hand, this argument seems to not give any guarantee on the size of $U$ and $V \setminus U$, 
hence the objective function can be arbitrarily bad again. 
But a slight fix of this rounding argument can work. We only need a large ``core'' of nodes 
so that the remaining nodes still have a decently large distance to it.

\begin{lemma} \label{lem:LargeCoreRounding}
Suppose we have a set of nodes $A \subseteq V$ with $|A| \geq \alpha n$ and $\sum_{i \in V} d(i,A) \geq \beta \cdot S^*$. 
Then the best cut of the form $\{ i \in V \mid d(i,A) \leq r\}$ is a $\frac{2}{\alpha \beta}$-approximation. 
\end{lemma}
\begin{proof}
Recall that the optimum value of the objective function is $\frac{C^*}{S^* \cdot (n-S^*)}$. 
Suppose we sample $r \sim [0,1]$ and take $U := U(r) := \{ i \in V \mid d(i,A) \leq r\}$ as a random cut. 
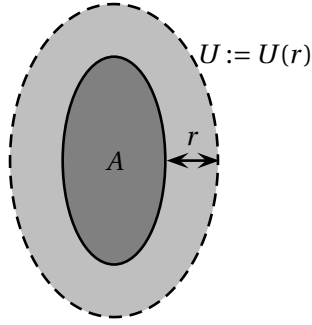
\begin{center}
\psset{unit=1.4cm}
\begin{pspicture}(-2,-1.5)(2,1.5)
\psellipse[fillcolor=lightgray,fillstyle=solid,linestyle=dashed](0,0)(1,1.5)
\psellipse[fillcolor=gray,fillstyle=solid](0,0)(0.5,1)
\pnode(0.5,0){A} \pnode(1.0,0){B}
\ncline{<->}{A}{B} \naput{$r$}
\rput[c](0,0){$A$}
\rput[l](0.8,1){$U := U(r)$}
\end{pspicture}
\end{center}
Then for an edge $(i,j) \in E$, say with $d(i,A) \leq d(j,A)$, we have 
\[
  \Pr_{r \sim [0,1]}[(i,j) \in \delta(U)] = \Pr_{r \sim [0,1]}[ d(i,A) \leq r \leq d(j,A) ] 
  \stackrel{\textrm{triangle inequality}}{\leq} d(i,j).
\]
Hence the expected cost of the cut $U$ is
\[
  \E_{r \sim [0,1]}[c(\delta(U))] \leq \sum_{(i,j) \in E} d(i,j) = C^*.
\]
Note that in any case $|U| \geq |A| \geq \alpha n$. We know that
$d(i,A) \leq 2$ for all $i \in V$ and hence $\Pr_{r \sim [0,1]}[i \notin U] = \Pr_{r \sim [0,1]}[r < d(i,A)] \geq \frac{1}{2}d(i,A)$.
Hence 
\[
  \E_{r \sim [0,1]}[|V \setminus U|] \geq \frac{1}{2} \sum_{i \in V} d(i,A) \geq \frac{\beta}{2} \cdot S^*.  
\]
Then $\E_{r \sim [0,1]}[|U| \cdot |V \setminus U|] \geq \frac{\alpha \beta}{2} \cdot S^* \cdot n$.
In other words, the random cut seems to have the right expected nominator and denominator to satisfy the claim. 
But this is \emph{not} enough to argue that their ratio satisfies $\E_{r \sim [0,1]}[\frac{c(\delta(U))}{|U| \cdot |V \setminus U|}] \leq \frac{2}{\alpha \beta} \cdot \frac{C^*}{S^* \cdot (n-S^*)}$. The following insight comes to rescue: 
\begin{quote}
{\bf Fact.} Let $\bm{a},\bm{b} \in \setR_{\geq 0}^m$ be non-negative numbers and $\pazocal{D}$ be a distribution over 
indices in $[m]$. Then 
\[
  \min_{i\in \{ 1,\ldots,m\}} \left\{ \frac{\bm{a}_i}{\bm{b}_i} \right\} \leq \frac{\E_{i \sim \pazocal{D}}[\bm{a}_i]}{\E_{i \sim \pazocal{D}}[\bm{b}_i]}.
\]
\end{quote}
Now, this fact implies that \emph{best} choice of $U$ (over all $r \in [0,1]$) will indeed satisfy the claim and the
lemma is proven.
\end{proof}
We should also remark that if $A$ is given, we can find the cut $U$ in polynomial time
as we only need to try out at most $n^2$ many values of $r$.

\section{The case of heavy clusters}

Let $B(i,r) := \{ j \in V \mid d(i,j) \leq r\}$ be the ``ball'' of radius $r$ around node $i$. 
A slight annoyance of the ARV algorithm is that it requires a case split. If we can find a 
\emph{cluster center} $i^* \in V$, 
then we can use Lemma~\ref{lem:LargeCoreRounding} to get a constant factor approximation
by just taking a ball around the center $i^*$. 
\begin{lemma}
Suppose there is a node $i^* \in V \cup \{0\}$ with $|B(i^*,\frac{1}{8} \cdot \frac{S^*}{n})| \geq \frac{n}{4}$. Then in polynomial time 
one can find a cut that gives a $O(1)$-approximation. 
\end{lemma}
\begin{proof}
We set $A := B(i^*, \frac{1}{8} \cdot \frac{S^*}{n})$. 
Then by assumption $|A| \geq \frac{n}{4}$. 
Moreover, bounding the average distance of pairs of nodes from above and from below gives
\[
 \frac{1}{2} \cdot \frac{S^*}{n} \leq \frac{S^*}{n} \cdot \underbrace{\frac{n-S^*}{n}}_{\geq 1/2} \stackrel{(II)}{\leq} \E_{i,j \sim V}[d(i,j)] \stackrel{\textrm{triangle ineq.}}{\leq} 2\cdot \E_{i \sim V} \Big(d(i,A)+\frac{1}{8} \cdot \frac{S^*}{n}\Big).
\]
This can be rearranged to $\E_{i \sim V}[d(i,A)] \geq \frac{1}{8} \cdot \frac{S^*}{n}$.
We obtain a $64$-approximation by applying Lemma~\ref{lem:LargeCoreRounding}. 
\end{proof}

\section{An algorithm for the main case}

From now on we make the assumption that no cluster exists: 
\[
  \left|B\Big(i,\frac{1}{8} \cdot \frac{S^*}{n}\Big)\right| < \frac{n}{4} \quad \forall i \in V \cup \{ 0\}.
\]
We will prove that in this case, there are sets $L,R \subseteq V$ of size $|L|,|R| \geq \Omega(n)$
with $d(L,R) \geq \Delta \cdot \frac{S^*}{n}$ for $\Delta := \Theta(1/\sqrt{\log n})$. Then choosing Lemma~\ref{lem:LargeCoreRounding} with $A := L$ 
will give a $O(\frac{1}{\Delta})$-approximation. 
Before we start proving this, we want to further simplify the situation. 
Note that by $(I)$ we have $\E_{i \sim V}[d(i,\bm{0})] = \frac{S^*}{n}$, and hence at most half the nodes can have a distance of 
more than $2 \cdot \frac{S^*}{n}$ to $\bm{0}$. Moreover we have $|B(\bm{0},\frac{1}{8} \cdot \frac{S^*}{n})| \leq \frac{n}{4}$. 
Then we only loose a constant factor if we delete
those nodes and assume that $\frac{1}{8} \cdot \frac{S^*}{n} \leq d(i,\bm{0}) \leq 2 \cdot \frac{S^*}{n}$ for all remaining nodes. 
Next we scale the vectors $\bm{v}_i$ by a factor
of $\sqrt{n / S^*}$, which scales the distances $d(i,j)$ by a factor of $\frac{n}{S^*}$. 
After this transformation it suffices to prove the following structure theorem: 

\begin{theorem}[ARV Structure Theorem]
Given any set of $|V| = n$ vectors $\{\bm{v}_i\}_{i \in V} \subseteq \setR^m$ with $\frac{1}{8} \leq \|\bm{v}_i\|_2^2 \leq 2$ and 
$|B(i,\frac{1}{8})| \leq \frac{3}{4}n$ for all $i \in V$ that satisfy the triangle inequalities
\[
  \|\bm{v}_i - \bm{v}_j\|_2^2 \leq \|\bm{v}_i - \bm{v}_k\|_2^2 + \|\bm{v}_k - \bm{v}_j\|_2^2  \quad \forall i,j,k \in V.
\]
Then there is a polynomial time algorithm that with constant
probability finds sets $L,R \subseteq V$ of size $|L|,|R| \geq \Omega(n)$ with $d(L,R) \geq \Delta$
for $\Delta := \Theta(1/\sqrt{\log n})$.
\end{theorem}
From now on, we complete ignore the cost function and only use the properties given in 
the Structure Theorem.
Such sets $L$ and $R$ with $d(L,R) \geq \Delta$ are called \emph{$\Delta$-separated}. 
Let $c'>0$ be a small enough constants. 
The algorithm to produce the $\Delta$-separated sets is as follows:

\begin{center}
\psframebox{%
\begin{minipage}{15cm}
\noindent {\bf Well-separated sets algorithm}  \vspace{1mm} \hrule \vspace{2mm}  
\noindent {\bf Input: } Vectors $\{\bm{v}_i\}_{i \in V} \subseteq \setR^m$ satisfying the triangle inequality with $\frac{1}{8} \leq \|\bm{v}_i\|_2^2 \leq 2$ and $|B(i,\frac{1}{8})| \leq \frac{3}{4}n$ for all $i \in V \cup \{0\}$.   \vspace{0.5mm} \\
\noindent {\bf Output: } Either $\Delta$-separated sets $L',R'$ of size $|L'|,|R'| \geq \frac{c'}{2}n$ or $\texttt{FAIL}$  \vspace{1mm} \hrule  \vspace{1mm}
\begin{enumerate*}
\item[(1)] Pick a random Gaussian $\bm{g} \sim N^m(0,1)$
\item[(2)] Set $L := \{ i \in V \mid \left<\bm{v}_i,\bm{g}\right> \leq -1\}$ and $R := \{ i \in V \mid \left<\bm{v}_i,\bm{g}\right> \geq 1\}$
\item[(3)] If either $|L| \leq c'n$ or $|R| \leq c'n$ then  $\texttt{FAIL}$
\item[(4)] Compute any inclusion-wise maximal matching $M(\bm{g}) \subseteq \{ (i,j) \in L \times R \mid d(i,j) \leq \Delta\}$
\item[(5)] If $|M(\bm{g})| > \frac{c'}{2}n$ then  $\texttt{FAIL}$
\item[(6)] Return $L' := \{ i \in L \mid i\textrm{ not covered by }M(\bm{g})\}$ and $R' := \{ i \in V \mid i \textrm{ not covered by } M(\bm{g})\}$
\end{enumerate*}
\end{minipage}
}
\end{center}
Observe that any pair $i \in L'$ and $j \in R'$ that remains will have $d(i,j) > \Delta$
as otherwise the matching $M(\bm{g})$ would not have been maximal. Also, if the algorithm 
reaches (6), then $\min\{ |L'|,|R'| \} \geq \frac{c'}{2}n$.
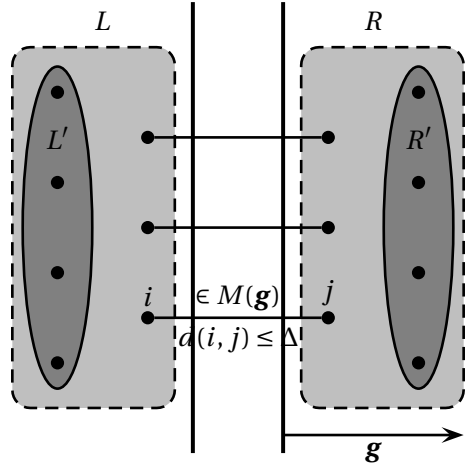
\begin{center}
\psset{unit=1.2cm}
\begin{pspicture}(-2,-2.5)(2,2.5)
\drawRect{linearc=0.2,linestyle=dashed,fillstyle=solid,fillcolor=lightgray}{-2.5}{-2}{1.8}{4}
\drawRect{linearc=0.2,linestyle=dashed,fillstyle=solid,fillcolor=lightgray}{0.7}{-2}{1.8}{4}
\psellipse[fillstyle=solid,fillcolor=gray](-2,0)(0.4,1.8)
\psellipse[fillstyle=solid,fillcolor=gray]( 2,0)(0.4,1.8)
\cnode*(-1, 1){2.5pt}{ML1}
\cnode*(-1, 0){2.5pt}{ML2}
\cnode*(-1,-1){2.5pt}{ML3}
\cnode*(-2, 1.5){2.5pt}{L1}
\cnode*(-2, 0.5){2.5pt}{L2}
\cnode*(-2,-0.5){2.5pt}{L3}
\cnode*(-2,-1.5){2.5pt}{L4}
\cnode*( 1, 1){2.5pt}{MR1}
\cnode*( 1, 0){2.5pt}{MR2}
\cnode*( 1,-1){2.5pt}{MR3}
\cnode*( 2, 1.5){2.5pt}{R1}
\cnode*( 2, 0.5){2.5pt}{R2}
\cnode*( 2,-0.5){2.5pt}{R3}
\cnode*( 2,-1.5){2.5pt}{R4}
\psline[linewidth=1.5pt](-0.5,-2.5)(-0.5,2.5)
\psline[linewidth=1.5pt]( 0.5,-2.5)( 0.5,2.5)
\ncline{ML1}{MR1}
\ncline{ML2}{MR2}
\ncline{ML3}{MR3}
\nput[labelsep=2pt]{90}{ML3}{$i$}
\nput[labelsep=2pt]{90}{MR3}{$j$}
\naput[labelsep=2pt]{$\in M(\bm{g})$}
\nbput[labelsep=2pt]{$d(i,j) \leq \Delta$}
\rput[c](-1.5,2.3){$L$}
\rput[c]( 1.5,2.3){$R$}
\rput[c](-2,1){$L'$}
\rput[c]( 2,1){$R'$}
\pnode(0.5,-2.3){A} \pnode(2.5,-2.3){B} \ncline{->}{A}{B} \nbput[labelsep=2pt]{$\bm{g}$}
\end{pspicture}
\end{center}

The first step is to argue that with constant probability both $L$ and $R$ have size at least $\Omega(n)$.
\begin{lemma}
There is an absolute constant $c' >0$ so that $\Pr_{\bm{g} \sim N^m(0,1)}[\min\{ |L|,|R|\} \geq c'n] \geq c'$.
\end{lemma}
\begin{proof}
We will prove that $\E[|L| \cdot |R|] \geq \Omega(n^2)$ which then implies the claim.
Fix any $i \in V$ and select one of the at least $\frac{1}{4}n$ nodes $j$ with $d(i,j) \geq \frac{1}{8}$. 
Let $\bm{w}$ be the orthogonal projection of $\bm{v}_j$ on $\bm{v}_i$, see figure below. 
Let $\alpha \in [0,\frac{\pi}{2}]$ be the angle between $\bm{v}_i-\bm{v}_j$ and $\bm{w}$ and let $\beta \in [0,\frac{\pi}{2}]$ be the angle between $\bm{v}_j$ and $\bm{w}$. 
Due to the triangle inequalities, the angle spanned by the points $\bm{0}$, $\bm{v}_i$ and $\bm{v}_j$ is non-obtuse
and $\alpha+\beta \leq \frac{\pi}{2}$. 
Then we have either $\alpha \leq \frac{\pi}{4}$ and
\[
  \|\bm{w}\|_2 = \underbrace{\cos(\alpha)}_{\geq 1/\sqrt{2}} \cdot \underbrace{\|\bm{v}_i - \bm{v}_j\|_2}_{\geq 1/\sqrt{8}} \geq \frac{1}{4}
\]
or otherwise we have $\beta \leq \frac{\pi}{4}$ and
\[
  \|\bm{w}\|_2 = \underbrace{\cos(\beta)}_{\geq 1/\sqrt{2}} \cdot \underbrace{\|\bm{v}_j\|_2}_{\geq 1/\sqrt{8}} \geq \frac{1}{4} 
\]
Either way, $\|\bm{w}\|_2 \geq \frac{1}{4}$.
Since $\bm{v}_i \perp \bm{w}$, the inner products $\left<\bm{g},\bm{v}_i\right>$ and 
$\left<\bm{g},\bm{w}\right>$ are independent random variables and we can estimate
\begin{eqnarray*}
  \Pr[i \in L\textrm{ and }j \in R] &\geq& \Pr[-2 \leq \left<\bm{g},\bm{v}_i\right> \leq -1\textrm{ and }\left<\bm{g},\bm{w}\right> \geq 3] \\
 &=& \Pr\Big[ -\underbrace{\frac{1}{\|\bm{v}_i\|_2}}_{\in [\frac{1}{\sqrt{8}},\frac{1}{\sqrt{2}} ]} \leq \left<\bm{g},\frac{\bm{v}_i}{\|\bm{v}_i\|_2}\right> \leq -\frac{2}{\|\bm{v}_i\|_2}\Big] \cdot \Pr\Big[\left<\bm{g},\frac{\bm{w}}{\|\bm{w}\|_2}\right> \geq \underbrace{\frac{3}{\|\bm{w}\|_2}}_{\leq 12}\Big]  > 0
\end{eqnarray*}
which is some tiny, yet absolute constant. 
Note that in case the latter event happens, then indeed 
\[
\left<\bm{g},\bm{v}_j\right> = \underbrace{\left<\bm{g},\bm{v}_i\right>}_{\geq -2} \cdot \underbrace{\frac{\left<\bm{v}_i,\bm{v}_j\right>}{\|\bm{v}_i\|_2^2}}_{\in [0,1]} + \underbrace{\left<\bm{g},\bm{w}\right>}_{\geq 3} \geq 1.
\] 
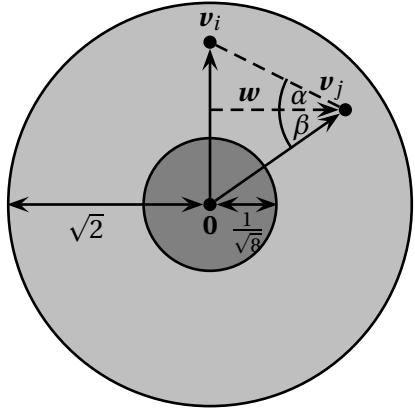
\begin{center}
\psset{unit=1.8cm}
\begin{pspicture}(-1,-1.3)(1,1.5)
\pscircle[fillstyle=solid,fillcolor=lightgray](0,0){1.5}
\pscircle[fillstyle=solid,fillcolor=gray](0,0){0.5}
\cnode*(0,0){2.5pt}{origin} \nput[labelsep=2pt]{-90}{origin}{$\bm{0}$}
\SpecialCoor
\cnode*(0,1.2){2.5pt}{vi} \nput[labelsep=2pt]{90}{vi}{$\bm{v}_i$}
\cnode*(1,0.7){2.5pt}{vj} \nput[labelsep=2pt]{120}{vj}{$\bm{v}_j$}
\ncline{->}{origin}{vi}
\ncline{->}{origin}{vj}
\pnode(0,0.7){w}
\ncline[linestyle=dashed]{->}{w}{vj} \naput[labelsep=2pt,npos=0.3]{$\bm{w}$}
\ncline[linestyle=dashed]{vj}{vi}
\psarc(vj){25pt}{180}{-144}
\psarc(vj){25pt}{152}{180}
\nput[labelsep=12pt]{165}{vj}{$\alpha$}
\nput[labelsep=12pt]{-160}{vj}{$\beta$}
\pnode(0.5,0){R1} \ncline{<->}{origin}{R1}\nbput[labelsep=2pt]{$\frac{1}{\sqrt{8}}$}
\pnode(-1.5,0){R2} \ncline{<->}{origin}{R2}\naput[labelsep=2pt,npos=0.6]{$\sqrt{2}$}
\end{pspicture}
\end{center}
\end{proof}
This implies that with a constant probability, the algorithm does not fail in (3). 
The main technical part lies in proving that $\E_{\bm{g} \sim N^m(0,1)}[|M(\bm{g})|] \leq cn$
 where we can make the constant $c$ as small as we want, at the expense of a smaller value of $\Delta$. 
If we choose $c := (\frac{c'}{2})^2$, then $\Pr[|M(\bm{g})| > \frac{c'}{2}] \leq \frac{c'}{2}$
and the success probability of the algorithm is at least $\frac{c'}{2}$.

\section{The proof of the Structure Theorem}

The following geometric theorem by Arora, Rao and Vazirani is the heart of their
$O(\sqrt{\log n})$-approximation for Sparsest Cut. To be precise, the original ARV result~\cite{DBLP:conf/stoc/AroraRV04}
only showed this theorem for $\Delta = \Theta((\log n)^{-2/3})$ and needed a lot of extra work
to get the $O(\sqrt{\log n})$-approximation. The claim as it is stated here was first proven by Lee~\cite{DBLP:conf/soda/Lee05}.
For an edge set $E'$, let $\beta(E')$ be the size of the maximum matching.
\begin{theorem}[\cite{DBLP:conf/stoc/AroraRV04,DBLP:conf/soda/Lee05}] \label{thm:StructureTheoremBoundOnMatchingSize} 
For any constant $c>0$ there is a choice of $\Delta := \Theta_c(1 / \sqrt{\log n})$ so that the following holds: 
Let $\{ \bm{v}_i\}_{i \in V} \subseteq \setR^m$ be a set of $|V| = n$ vectors satisfying the triangle inequality
\[
 \|\bm{v}_i - \bm{v}_j\|_2^2 \leq \|\bm{v}_i - \bm{v}_k\|_2^2 + \|\bm{v}_k - \bm{v}_j\|_2^2 \quad  \forall i,j,k \in V.
\]
For a vector  $\bm{g} \in \setR^m$ define
\[
  E(\bm{g}) := \left\{ (i,j) \in V \times V \mid \left<\bm{v}_j - \bm{v}_i,\bm{g}\right> \geq 2\textrm{ and } \|\bm{v}_i - \bm{v}_j\|_2^2 \leq \Delta \right\}.
\]
Then $\E_{\bm{g} \sim N^m(0,1)}[\beta(E(\bm{g}))] \leq cn$.
\end{theorem}
Here we think of $E(\bm{g})$ as directed edges. 
Let $M(\bm{g})$ be a maximum matching attaining $\beta(E(\bm{g}))$.
We will assume the existence of such a matching $M(\bm{g})$ and lead this to a contradiction. 
By inducing on a subgraph and reducing the constant $c$ one can even assume that for \emph{every} node
the probability of having an outgoing edge is at least $c$ and the same is true for ingoing edges. 
First, there is no harm in assuming that $M(\bm{g})$ has the reverse edges of $M(-\bm{g})$,
which implies that each node has an outgoing edge with the same probability as it has an incoming edge. 
\begin{lemma}
Assume that Theorem~\ref{thm:StructureTheoremBoundOnMatchingSize} is false for vectors $\{ \bm{v}_i\}_{i\in V} \subseteq \setR^m$. 
Then there is a subset $V' \subseteq V$ of size $|V'| \geq cn$ and a matching $M'(\bm{g}) \subseteq M(\bm{g})$ on $V'$ so that 
every node $i \in V'$ has an outgoing edge in $M'(\bm{g})$ with probability at least $\frac{c}{8}$ and an ingoing edge with probability at least $\frac{c}{8}$. 
\end{lemma}
\begin{proof}
For each node $i \in V$ define $p(i) := \Pr_{\bm{g} \sim N^m(0,1)}[i\textrm{ has outgoing edge in }M(\bm{g})]$. 
If there is a node $i$ with $p(i) \leq \frac{c}{8}$, then we imagine to delete the node from the graph and remove
from $M(\bm{g})$ any edge containing node $i$. Note that this decreases the expected size of the matching 
by at most $2 \cdot \frac{c}{8}$. 
We continue this procedure until no such node exists anymore. Let $V'$ be the remaining set of nodes
with $M'(\bm{g}) := M(\bm{g}) \cap (V' \times V')$. Then $\E_{\bm{g} \sim N^m(0,1)}[|M'(\bm{g})|] \geq \E_{\bm{g} \sim N^m(0,1)}[|M(\bm{g})|] - n \cdot \frac{c}{4} \geq \frac{c}{2}n$. 
Then there must be at least $|V'| \geq cn$ many nodes left.  
\end{proof}
After changing the constants and adapting the value of $n$, we assume to have $n$ nodes and 
for every node $i \in V$, the matching $M(\bm{g})$ has an outgoing and an incoming edge with 
probability at least $c$. 

We call an edge $(i,j)$ \emph{$\Delta$-short} if $d(i,j) \leq \Delta$. 
For a node $i \in V$, let $\Gamma(i) := \{ j \in V \mid d(i,j) \leq \Delta\}$ be the \emph{neighborhood} of $i$
with respect to the graph of $\Delta$-short edges. Moreover, let $\Gamma_k(i) := \Gamma_{k-1}(\Gamma(i))$
be the nodes that can be reached from $i$ via at most $k$ many $\Delta$-short edges. 
\begin{lemma} \label{lem:DistanceOfkNeighborhood}
For any $k \in \setZ_{\geq 0}$ and $i' \in \Gamma_k(i)$ one has $\|\bm{v}_i - \bm{v}_{i'}\|_2 \leq \sqrt{k\Delta}$.
\end{lemma}
\begin{proof}
We have $\|\bm{v}_i - \bm{v}_{i'}\|_2^2 = d(i,i') \leq k \cdot \Delta$ by the SDP triangle inequality. Taking square roots
gives the claim.
\end{proof}
At the heart of the arguments lies the fact that the value of Lipschitz functions is well concentrated.
Recall that a function $F : \setR^m \to \setR$ is called \emph{$L$-Lipschitz}, if
$|F(\bm{x}) - F(\bm{y})| \leq L \cdot \|\bm{x} - \bm{y}\|_2$ for all $\bm{x},\bm{y} \in \setR^m$. 
\begin{lemma}[Concentration for Lipschitz Functions (Sudakov-Tsirelson, Borrell)] \label{lem:ConcentrationForLipschitzFunctions}
Let $F : \setR^m \to \setR$ be an $L$-Lipschitz function with 
Gaussian mean $\mu := \E_{\bm{g} \sim N^m(0,1)}[F(\bm{g})]$. Then $\Pr_{\bm{g} \sim N^m(0,1)}[|F(\bm{g}) - \mu| \geq \alpha] \leq 2e^{-\frac{1}{2} \cdot (\alpha/L)^2}$ for any $\alpha \geq 0$.
\end{lemma}

We define a function
\[
  F_{i,k}(\bm{g}) := \max\{ \left<\bm{g},\bm{v}_i - \bm{v}_j\right> \mid j \in \Gamma_k(i) \}.
\]
In other words, $F_{i,k}(\bm{g})$ gives the maximum inner product $\left<\bm{g},\bm{v}_i-\bm{v}_j\right>$ over all nodes $j \in V$ that are 
within $k$ many $\Delta$-short edges of node $i$. 
Note that $F_{i,k}(\bm{g}) \geq \left<\bm{g},\bm{v}_i-\bm{v}_i\right> = 0$ for all $\bm{g} \in \setR^m$ as $i \in \Gamma_k(i)$.
\begin{lemma} \label{lem:LipschitzPropertyOfFik}
The function $F_{i,k} : \setR^m \to \setR$ is $\sqrt{k \cdot \Delta}$-Lipschitz. 
\end{lemma} 
\begin{proof}
Fix $\bm{g},\bm{g}' \in \setR^m$ and assume for the sake of symmetry that $F(\bm{g}) \geq F(\bm{g}')$.
Let $j,j' \in \Gamma_k(i)$ be the nodes attaining $F(\bm{g})$ and $F(\bm{g}')$, resp. 
Then
\[
  |F(\bm{g}) - F(\bm{g}')| = \left<\bm{g},\bm{v}_{i}-\bm{v}_{j}\right> - \left<\bm{g}',\bm{v}_{i} - \bm{v}_{j'}\right> 
 \leq \left<\bm{g}-\bm{g}',\bm{v}_{i}-\bm{v}_j\right>  
\stackrel{\textrm{Cauchy-Schwarz}}{\leq} \|\bm{g} - \bm{g}'\|_2 \cdot \underbrace{\|\bm{v}_{i} - \bm{v}_j\|_2}_{\leq \sqrt{k \cdot \Delta}\textrm{ by Lem.~\ref{lem:DistanceOfkNeighborhood}}}. 
\]
Here we used in the first inequality that $j' \in \Gamma_k(i)$ maximizes $\left<\bm{g}',\bm{v}_{i}-\bm{v}_{j'}\right>$.
\end{proof}

Now, let $\mu_{i,k} := \E_{\bm{g} \sim N^m(0,1)}[F_{i,k}(\bm{g})]$ be the expected maximum inner product over $k$-neighbors
of node $i \in V$.
One useful argument will be a nice relation between expectations of neighbors: 
\begin{lemma} \label{lem:BehauviorOfMuKIForNeighbors}
For any node $i \in V$ and $i' \in \Gamma(i)$ and $k \in \setZ_{\geq 0}$ one has $\mu_{i',k+1} \geq \mu_{i,k}$.
\end{lemma}
\begin{proof}
We have
\begin{eqnarray*}
 \mu_{i',k+1} &=& \E_{\bm{g} \sim N^m(0,1)}\Big[\max \Big\{ \left<\bm{v}_{i'} - \bm{v}_{i},\bm{g}\right> + \left<\bm{v}_i - \bm{v}_{j},\bm{g}\right>\mid j \in \Gamma_{k+1}(i') \Big\}\Big] \\
 &\stackrel{\Gamma_{k+1}(i') \supseteq \Gamma_k(i)}{\geq}& \E_{\bm{g} \sim N^m(0,1)}\left[\max\left\{ \left<\bm{v}_i-\bm{v}_{j},\bm{g}\right> \mid j \in \Gamma_k(i) \right\} \right] + \underbrace{\E_{\bm{g} \sim N^m(0,1)}[\left<\bm{v}_{i'} - \bm{v}_{i},\bm{g}\right>]}_{=0} = \mu_{i,k}.
\end{eqnarray*}
\end{proof}
On the other hand, we can get the following upper bound: 
\begin{lemma} \label{lem:UpperBoundOnMuKI}
For any $k \in \setZ_{\geq 0}$ and any $i \in V$ one has $\mu_{i,k} \leq 10\sqrt{\log n} \cdot \sqrt{k \Delta}$. 
\end{lemma}
\begin{proof}
For any $j \in \Gamma_k(i)$ we have $\|\bm{v}_i - \bm{v}_j\|_2 \leq \sqrt{k\Delta}$ and generously
\[
\Pr_{\bm{g} \sim N^m(0,1)}\Big[\left<\bm{v}_i-\bm{v}_j,\bm{g}\right> \geq 8\sqrt{\log n} \cdot \sqrt{k\Delta}\Big] \leq \int_{8\sqrt{\log n}}^{\infty} \frac{1}{\sqrt{2\pi}}e^{-x^2/2} dx \leq \frac{1}{2n}.
\] 
That means $\Pr_{\bm{g} \sim N^m(0,1)}[F_{i,k}(\bm{g}) \leq 8\sqrt{\log n} \cdot \sqrt{k\Delta}] \geq 1/2$.  Again, $F_{i,k}$ is $\sqrt{k\Delta}$-Lipschitz,
hence $\Pr_{\bm{g} \sim N^m(0,1)}[|F_{i,k}(\bm{g}) - \mu_{i,k}| \geq 2\sqrt{k\Delta}] \leq \frac{1}{2}$ and $2 \leq 2\sqrt{\log(n)}$ for $n \geq 2$.
\end{proof}

\subsection{Extend or expand}

The core argument to get to a contradiction is the following: 
\begin{lemma} \label{lem:ExpansionOrExtensionLemma}
Let $\Delta>0$, $\delta \in \setR$ and $k \in \{ 0,\ldots,\frac{1}{100} \log(1/c) \cdot \frac{1}{\Delta} \}$ be parameters
and $U \subseteq \{ i \in V \mid \mu_{i,k} \geq \delta \}$ be a set of nodes. Then 
\begin{enumerate}
\item[(A)] either there is a subset $U' \subseteq \Gamma(U)$ so that $|U'| \geq \frac{c}{4} \cdot |U|$ and $\mu_{i,k+1} \geq \delta+1$ for all $i \in U'$. 
\item[(B)] or the neighborhood $U' := \Gamma(U)$ satisfies $|U'| \geq \frac{4}{c} \cdot |U|$ and
$\mu_{i,k+1} \geq \delta$ for all $i \in U'$. 
\end{enumerate}
\end{lemma}

\begin{proof}
If $|\Gamma(U)| \geq \frac{4}{c} \cdot |U|$, then every node in $i \in \Gamma(U)$ has $\mu_{i,k+1} \geq \delta$ by Lemma~\ref{lem:BehauviorOfMuKIForNeighbors}
and we are in case (B). 
So suppose that $|\Gamma(U)| < \frac{4}{c} \cdot |U|$. Consider the random matching
\[
  \tilde{M}(\bm{g}) := \Big\{ (i,j) \in M(\bm{g}) \mid i \in U\textrm{ and } F_{i,k}(\bm{g}) \geq \delta-\frac{1}{2} \Big\}
\] 
that is the restriction of $M(\bm{g})$ to edges that are going out of $U$ and where  $F_{i,k}(\bm{g})$ is
large enough. 
Note that $\mu_{i,k} = \E_{\bm{g} \sim N^m(0,1)}[F_{i,k}(\bm{g})] \geq \delta$ for all $i \in U$ and
$F_{i,k}$ is $\sqrt{k\Delta}$-Lipschitz. Hence by Lemma~\ref{lem:LipschitzPropertyOfFik} we have
\[
 \Pr_{\bm{g} \sim N^m(0,1)}\Big[F_{i,k}(\bm{g}) < \delta - \frac{1}{2}\Big] \leq 2\exp\Big(-\frac{1}{8k\Delta}\Big) < \frac{c}{2} \quad \forall i \in U.
\]
This implies that each node in $U$ will have an outgoing edge in $\tilde{M}(\bm{g})$
with probability at least $\frac{c}{2}$. 
Now, define $U' := \{ j \in \Gamma(U) \mid \Pr[j\textrm{ has incoming edge from }\tilde{M}(\bm{g})] \geq \frac{c^2}{16} \}$. 
Since $\tilde{M}(\bm{g})$ is a matching we have
\[
\frac{c}{2} \cdot |U|  \leq \E[|\tilde{M}(\bm{g})|] \leq \frac{c^2}{16} \cdot \underbrace{|\Gamma(U) \setminus U'|}_{\leq (4/c) \cdot |U|} + |U'| \leq \frac{c}{4} \cdot |U| + |U'| 
\]
which implies that $|U'| \geq \frac{c}{4} \cdot |U|$. Now fix a node $j \in U'$. 
It remains to argue that $\mu_{j,k+1} \geq \delta+1$ for all $j \in U'$. 
First, condition on the event that $\tilde{M}(\bm{g})$ has an edge incoming to $j$.
We denote that edge by $(i(\bm{g}),j) \in \tilde{M}(\bm{g})$
with $i(\bm{g}) \in U$\footnote{Here we write $i(\bm{g})$ to indicate that this node will depend on the choice of $\bm{g}$.}. 
We know by the definition of $\tilde{M}(\bm{g})$ that $\left<\bm{v}_{j} - \bm{v}_{i(\bm{g})},\bm{g}\right> \geq 2$. Moreover, 
we know that there is a node $h(\bm{g}) \in \Gamma_{k}(i(\bm{g}))$ so that $\left<\bm{v}_{i(\bm{g})} - \bm{v}_{h(\bm{g})},\bm{g} \right> \geq \delta - \frac{1}{2}$. 
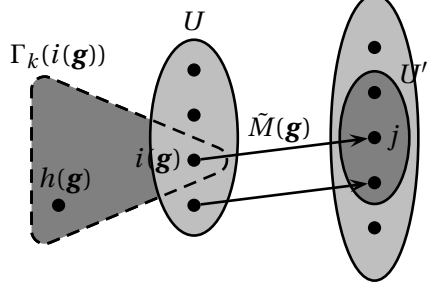
\begin{center}
\psset{unit=1.2cm}
\begin{pspicture}(-2,-1.6)(2,1.6)
\pspolygon[linestyle=dashed,linearc=5pt,fillcolor=gray,fillstyle=solid](0.6,-0.25)(-1.8,-1.25)(-1.8,0.75)
\psellipse[fillcolor=lightgray,fillstyle=solid](0,0)(0.5,1.1)
\psellipse[fillcolor=lightgray,fillstyle=solid](2,0)(0.5,1.6)
\psellipse[fillcolor=gray,fillstyle=solid](2,0)(0.4,0.75)
\pspolygon[linestyle=dashed,linearc=5pt,fillcolor=gray,fillstyle=none](0.6,-0.25)(-1.8,-1.25)(-1.8,0.75)
\cnode*(0,0.75){2.5pt}{v1}
\cnode*(0,0.25){2.5pt}{v2}
\cnode*(0,-0.25){2.5pt}{v3}
\cnode*(0,-0.75){2.5pt}{v4}
\cnode*(2,1.0){2.5pt}{w1}
\cnode*(2,0.5){2.5pt}{w2}
\cnode*(2,0){2.5pt}{w3}
\cnode*(2,-0.5){2.5pt}{w4}
\cnode*(2,-1.0){2.5pt}{w5}
\cnode*(-1.5,-0.75){2.5pt}{u1}
\rput[c](-1.5,0.9){$\Gamma_k(i(\bm{g}))$}
\nput[labelsep=2pt]{80}{u1}{$h(\bm{g})$}
\nput[labelsep=2pt]{180}{v3}{$i(\bm{g})$}
\nput[labelsep=2pt]{0}{w3}{$j$}
\ncline[arrowsize=5pt]{->}{v3}{w3} \naput[labelsep=0pt]{$\tilde{M}(\bm{g})$}
\ncline[arrowsize=5pt]{->}{v4}{w4}
\rput[c](0,1.3){$U$}
\rput[c](2,1.75){$\Gamma(U)$}
\nput[labelsep=8pt]{30}{w2}{$U'$}
\end{pspicture}
\end{center}
Then 
$\left<\bm{v}_{j} - \bm{v}_{h(\bm{g})},\bm{g}\right> = \left<\bm{v}_{j} - \bm{v}_{i(\bm{g})},\bm{g}\right> + \left<\bm{v}_{i(\bm{g})} - \bm{v}_{h(\bm{g})},\bm{g}\right> \geq \delta + \frac{3}{2}$. 
In other words, $\Pr_{\bm{g} \sim N^m(0,1)}[F_{j,k+1}(\bm{g}) \geq \delta + \frac{3}{2}] \geq \frac{c^2}{16}$. 
Again, the function $F_{j,k+1}$ is $\sqrt{(k+1)\Delta}$-Lipschitz and 
\[
  \Pr_{\bm{g} \sim N^m(0,1)}\Big[|F_{j,k+1}(\bm{g}) - \mu_{j,k+1}| \geq \frac{1}{2}\Big] \stackrel{\textrm{Lem.~\ref{lem:ConcentrationForLipschitzFunctions}}}{\leq} 2\exp\Big(-\frac{1}{8(k+1)\Delta}\Big) < \frac{c^2}{16}
\]
and hence $\mu_{j,k+1} \geq \delta+1$. This shows the claim. 
\end{proof}

Now, suppose we run Lemma~\ref{lem:ExpansionOrExtensionLemma} iteratively, starting with $U := V$
and in each iteration we replace the current $U$ by the set $U'$. 
We iterate this until the upper bound on $k$ is reached. Note that case (B) cannot happen more often than case (A) 
as always $|\Gamma(U)| \leq n$. 
Then after being $k = \frac{1}{100\Delta}\log(1/c)$ times in Case (A) and  $\ell \in \{ 0,\ldots k\}$ times in Case (B), 
we end up with a set $U \subseteq V$ with $|U| \geq n \cdot (\frac{c}{4})^{k-\ell} \geq n \cdot (\frac{c}{4})^{k}$ and 
$\mu_{i,2k} \geq \mu_{i,k+\ell} \geq k$ for all $ i \in U$. 
On the other hand, $\mu_{i,2k} \leq 10\sqrt{\log n} \cdot \sqrt{2k\Delta}$ by Lemma~\ref{lem:UpperBoundOnMuKI}. 
Choosing $\Delta := \Theta_c(\frac{1}{\sqrt{\log n}})$ and $k := \Theta_c(\sqrt{\log n})$ 
with proper choice of constants, then gives a contradiction. 

\paragraph{Acknowledgement.}

The author is very grateful to James R. Lee, Harishchandra Ramadas, Rebecca Hoberg and Alireza Rezaei for helpful discussion and comments. 

\bibliographystyle{alpha}
\bibliography{LectureNotesOnARV}

\begin{thebibliography}{ARV04}

\bibitem[ARV04]{DBLP:conf/stoc/AroraRV04}
Sanjeev Arora, Satish Rao, and Umesh~V. Vazirani.
\newblock Expander flows, geometric embeddings and graph partitioning.
\newblock In L{\'{a}}szl{\'{o}} Babai, editor, {\em Proceedings of the 36th
  Annual {ACM} Symposium on Theory of Computing, Chicago, IL, USA, June 13-16,
  2004}, pages 222--231. {ACM}, 2004.

\bibitem[GLS93]{GroetschelLovaszSchrijver88}
Martin Gr{\"o}tschel, L{\'a}szl{\'o} Lov{\'a}sz, and Alexander Schrijver.
\newblock {\em Geometric algorithms and combinatorial optimization}, volume~2
  of {\em Algorithms and Combinatorics}.
\newblock Springer-Verlag, Berlin, second edition, 1993.

\bibitem[GW94]{DBLP:conf/stoc/GoemansW94}
Michel~X. Goemans and David~P. Williamson.
\newblock .879-approximation algorithms for {MAX} {CUT} and {MAX} 2sat.
\newblock In Frank~Thomson Leighton and Michael~T. Goodrich, editors, {\em
  Proceedings of the Twenty-Sixth Annual {ACM} Symposium on Theory of
  Computing, 23-25 May 1994, Montr{\'{e}}al, Qu{\'{e}}bec, Canada}, pages
  422--431. {ACM}, 1994.

\bibitem[Lee05]{DBLP:conf/soda/Lee05}
James~R. Lee.
\newblock On distance scales, embeddings, and efficient relaxations of the cut
  cone.
\newblock In {\em Proceedings of the Sixteenth Annual {ACM-SIAM} Symposium on
  Discrete Algorithms, {SODA} 2005, Vancouver, British Columbia, Canada,
  January 23-25, 2005}, pages 92--101. {SIAM}, 2005.

\end{thebibliography}

\end{document}